\documentclass[11pt]{article}

\pdfoutput=1

\usepackage{fullpage}
\usepackage{amsfonts, amssymb, amsmath, amsthm}
\usepackage{latexsym}
\usepackage{url}
\usepackage{color}
\definecolor{DarkGray}{rgb}{0.1,0.1,0.5}
\usepackage[colorlinks=true,breaklinks, linkcolor= DarkGray,citecolor= DarkGray,urlcolor= DarkGray]{hyperref}	

\usepackage[tracking=smallcaps]{microtype}

\def\place #1#2#3{\mspace{#2}\makebox[0pt]{\raisebox{#3}{#1}}\mspace{-#2}}	

\newcommand{\bra}[1]{{\langle#1|}}
\newcommand{\ket}[1]{{|#1\rangle}}
\newcommand{\braket}[2]{{\langle#1|#2\rangle}}
\newcommand{\ketbra}[2]{{\ket{#1}\!\bra{#2}}}
\newcommand{\abs}[1]{{\lvert #1\rvert}}	
\newcommand{\norm}[1]{{\| #1 \|}}

\DeclareMathOperator{\Ex}{\operatorname{E}}

\def\C {{\bf C}}
\def\N {{\bf N}}

\def\L {{\mathcal L}}
\def\D {{\mathcal D}}

\newcommand{\identity}{\ensuremath{\boldsymbol{1}}} 

\newcommand{\ADV} {\mathrm{Adv}}
\newcommand{\ADVpm} {\mathrm{Adv}^{\pm}}
\def\Pim {{\overline{\Pi}}}

\DeclareMathOperator{\abst}{\operatorname{abs}}
\def\biadj {B}	

\newcommand{\comment}[1]{\emph{\color{blue}Comment:\color{black} #1}} 
\newlength{\commentslength}
\newcommand{\comments}[1]{
\hspace{-2\parindent}
\addtolength{\commentslength}{-\commentslength}
\addtolength{\commentslength}{\linewidth}
\addtolength{\commentslength}{-\parindent}
\fcolorbox{blue}{white}{\smallskip\begin{minipage}[c]{\commentslength}
\emph{Comments:}\begin{itemize}#1\end{itemize}\end{minipage}}\bigskip
}
\renewcommand{\comment}[1]{}\renewcommand{\comments}[1]{}
\newcommand{\rem}[1]{}

\newtheorem{theorem}{Theorem}[section]
\newtheorem{lemma}[theorem]{Lemma}
\newtheorem{corollary}[theorem]{Corollary}

\newtheorem{proposition}[theorem]{Proposition}

\newtheorem{definition}[theorem]{Definition}

\newfont{\subsubsecfnt}{ptmri8t at 10pt}
\renewcommand{\subparagraph}[1]{\smallskip{\subsubsecfnt #1.}}

\numberwithin{equation}{section} 

\newcommand{\eqnref}[1]{\hyperref[#1]{{(\ref*{#1})}}}
\newcommand{\thmref}[1]{\hyperref[#1]{{Theorem~\ref*{#1}}}}
\newcommand{\lemref}[1]{\hyperref[#1]{{Lemma~\ref*{#1}}}}
\newcommand{\corref}[1]{\hyperref[#1]{{Corollary~\ref*{#1}}}}
\newcommand{\defref}[1]{\hyperref[#1]{{Definition~\ref*{#1}}}}
\newcommand{\secref}[1]{\hyperref[#1]{{Section~\ref*{#1}}}}
\newcommand{\figref}[1]{\hyperref[#1]{{Figure~\ref*{#1}}}}
\newcommand{\tabref}[1]{\hyperref[#1]{{Table~\ref*{#1}}}}
\newcommand{\remref}[1]{\hyperref[#1]{{Remark~\ref*{#1}}}}
\newcommand{\appref}[1]{\hyperref[#1]{{Appendix~\ref*{#1}}}}
\newcommand{\claimref}[1]{\hyperref[#1]{{Claim~\ref*{#1}}}}
\newcommand{\propref}[1]{\hyperref[#1]{{Proposition~\ref*{#1}}}}
\newcommand{\exampleref}[1]{\hyperref[#1]{{Example~\ref*{#1}}}}
\newcommand{\conjref}[1]{\hyperref[#1]{{Conjecture~\ref*{#1}}}}

\allowdisplaybreaks[1]

\begin{document}

\title{Reflections for quantum query algorithms}
\author{Ben W.~Reichardt}
\date{}

\maketitle

\begin{abstract}
We show that any boolean function can be evaluated optimally by a quantum query algorithm that alternates a certain fixed, input-independent reflection with a second reflection that coherently queries the input string.  Originally introduced for solving the unstructured search problem, this two-reflections structure is therefore a universal feature of quantum algorithms.  

Our proof goes via the general adversary bound, a semi-definite program (SDP) that lower-bounds the quantum query complexity of a function.  By a quantum algorithm for evaluating span programs, this lower bound is known to be tight up to a sub-logarithmic factor.  The extra factor comes from converting a continuous-time query algorithm into a discrete-query algorithm.  We give a direct and simplified quantum algorithm based on the dual SDP, with a bounded-error query complexity that matches the general adversary bound.  

Therefore, the general adversary lower bound is tight; it is in fact an SDP for quantum query complexity.  This implies that the quantum query complexity of the composition $f \circ (g, \ldots, g)$ of two boolean functions $f$ and $g$ matches the product of the query complexities of $f$ and $g$, without a logarithmic factor for error reduction.  It further shows that span programs are equivalent to quantum query algorithms.  
\end{abstract}

\section{Introduction} \label{s:introduction}

The query complexity, or decision-tree complexity, of a function measures the number of input bits that must be read in order to evaluate the function.  Computation between queries is not counted.  Quantum algorithms can run in superposition, and the quantum query complexity therefore allows coherent access to the input string.  Quantum query complexity with bounded error lies below classical randomized query complexity, sometimes with a large gap~\cite{BernsteinVazirani97complexity, Simon97promise, Shor95factor, Aaronson09separation}, but for total functions~\cite{BealsBuhrmanCleveMoscaWolf98} or partial functions satisfying certain symmetries~\cite{AaronsonAmbainis09queries} the two measures are polynomially related; see the survey~\cite{BuhrmanDeWolf02querysurvey}.  

Although the query complexity of a function can fall well below its time complexity, studying query complexity has historically given insight into the power of quantum computers.  For example, the quantum part of Shor's algorithms for integer factorization and discrete logarithm is a quantum query algorithm for period finding~\cite{Shor95factor}.  Unlike for time complexity, there are also strong information-theoretic methods for placing lower bounds on quantum query complexity.  These lower-bound techniques can be broadly classified as using either the polynomial method~\cite{BealsBuhrmanCleveMoscaWolf98} or the adversary method~\cite{Ambainis00adversary, SpalekSzegedy04advequivalent}.  H{\o}yer and {\v S}palek~\cite{HoyerSpalek05adversarysurvey} have surveyed the development of these two techniques and their multitude of applications.  For now, suffice it to say that the two techniques are incomparable.  In particular, for the $n$-input collision problem, the best adversary lower bound is of~$O(1)$, whereas the correct complexity, determined by the polynomial method~\cite{AaronsonShi04collisioned} and a matching algorithm~\cite{BrassardHoyerTapp98collision} is $\Theta(n^{1/3})$.  

\clearpage

However, H{\o}yer, Lee and {\v S}palek~\cite{HoyerLeeSpalek07negativeadv} discovered a strict generalization of the adversary bound that remains a lower bound on quantum query complexity: 

\begin{definition}[Adversary bounds] \label{t:adversarydef}
For finite sets $C$ and $E$, and $\D \subseteq C^n$, let $f: \D \rightarrow E$.  An adversary matrix for $f$ is a $\abs{\D} \times \abs{\D}$ real, symmetric matrix $\Gamma$ that satisfies $\bra x \Gamma \ket y = 0$ for all~$x, y \in \D$ with $f(x) = f(y)$.  Define the adversary and general adversary bounds for $f$ by 
\begin{align}
\ADV(f) &= 
\max_{\Gamma \geq 0}
\frac{\norm{\Gamma}}{\max_{j \in [n]} \norm{\Gamma \circ \Delta_j}} \\
\ADVpm(f) &= 
\max_\Gamma
\frac{\norm{\Gamma}}{\max_{j \in [n]} \norm{\Gamma \circ \Delta_j}}
 \enspace .
\end{align}
Both maximizations are over adversary matrices $\Gamma$, required to be entry-wise nonnegative in $\ADV(f)$.  $\Gamma \circ \Delta_j$ denotes the entry-wise matrix product between $\Gamma$ and $\Delta_j = \sum_{x, y \in \D : x_j \neq y_j} \ketbra x y$.  
\end{definition}

Although the definitions of the two adversary bounds are very similar, the general adversary bound is much more powerful.  
In fact, the general adversary lower bound is always nearly tight: 

\begin{theorem}[\cite{Reichardt09spanprogram}] \label{t:querycomplexitytightnonbinary}
For any function $f : \D \rightarrow E$, with $\D \subseteq C^n$, the quantum query complexity $Q(f)$ satisfies 
\begin{equation} \label{e:querycomplexitytightnonbinary}
Q(f) = O\bigg(\ADVpm(f) \, \frac{\log \ADVpm(f)}{\log \log \ADVpm(f)} \log \abs{C} \cdot \log \abs{E} \bigg)
 \enspace .
\end{equation}
\end{theorem}

This surprising upper bound follows from a connection between quantum query algorithms and the span program computational model~\cite{KarchmerWigderson93span} first observed in~\cite{ReichardtSpalek08spanprogram} and significantly strengthened in~\cite{Reichardt09spanprogram, Reichardt09spanprogramfocs}.  Note that the original statement~\cite[Theorem~10.2]{Reichardt09spanprogram} of \thmref{t:querycomplexitytightnonbinary} restricted to the case $\abs C = 2$ and included an additional factor of $\log \log \abs E$---this factor can be removed by~\cite[Corollary~3]{BuhrmanNewmanRohrigDeWolf05robust}.  Lee has shown that the general adversary bound of a function with boolean output is stable under encoding the input into binary~\cite{Lee09adversaryboolean}, allowing the restriction $\abs C = 2$ to be removed at a logarithmic cost.  

\thmref{t:querycomplexitytightnonbinary} and the connection between span programs and quantum query algorithms behind its proof have corollaries including a query-optimal and nearly time-optimal quantum algorithm for evaluating a large class of read-once formulas over any finite gate set~\cite{Reichardt09unbalancedformula}.  However, is \thmref{t:querycomplexitytightnonbinary} optimal?  The factors of $\log \abs C$ and $\log \abs E$ are natural, but the $\log$ over $\log \log$ term is not.  It comes from converting a certain continuous-time query algorithm into a discrete-query algorithm following~\cite{CleveGottesmanMoscaSommaYongeMallo08discretize}.  This conversion also somewhat obscures the algorithm's structure.  

It was conjectured that the unnatural $\log$ over $\log \log$ factor could be removed~\cite[Conjecture~11.1]{Reichardt09spanprogram}.  In this article, we confirm the conjecture: 

\begin{theorem} \label{t:advboundtight}
For any function $f : \D \rightarrow \{0,1\}$, with $\D \subseteq \{0,1\}^n$, the general adversary bound characterizes quantum query complexity: 
\begin{equation} \label{e:advboundtight}
Q(f) = \Theta\big(\ADVpm(f)\big)
 \enspace .
\end{equation}
\end{theorem}

\thmref{t:advboundtight} suffices to simplify Eq.~\eqnref{e:querycomplexitytightnonbinary} following the proof of \cite[Theorem~10.2]{Reichardt09spanprogram}: 

\begin{corollary}
For finite sets $C$ and $E$, $\D \subseteq C^n$, and any function $f : \D \rightarrow E$, 
\begin{equation} \label{e:querycomplexityexactlytightnonbinary}
Q(f) = \Omega\big(\ADVpm(f)\big)
\qquad \text{and} \qquad
Q(f) = O\big(\ADVpm(f) \log \abs{C} \cdot \log \abs{E} \big)
 \enspace .
\end{equation}
\end{corollary}

\def\compose{\bullet}

\thmref{t:advboundtight} also allows for obtaining optimal results for the query complexity of composed~functions.  For $f: \{0,1\}^n \rightarrow \{0,1\}$ and $g: \{0,1\}^m \rightarrow \{0,1\}$, let $f \compose g$ be the function $\{0,1\}^{n m} \rightarrow \{0,1\}$ defined by $(f \compose g)(x) = f\big(g(x_1, \ldots, x_m), \ldots, g(x_{(n-1)m+1}, \ldots, x_{n m})\big)$.  A bounded-error algorithm for evaluating $f \compose g$ can be built from composing bounded-error algorithms for $f$ and $g$; thus, $Q(f \compose g) = O\big(Q(f) Q(g) \log Q(f)\big)$.  However, the logarithmic factor for error reduction can be removed, and there is a matching lower bound: 

\begin{theorem} \label{t:composition}
Let $f: \{0,1\}^n \rightarrow \{0,1\}$ and $g : \{0,1\}^m \rightarrow \{0,1\}$.  Then 
\begin{equation}
Q(f \compose g) = \Theta\big(Q(f) Q(g)\big)
 \enspace .
\end{equation}
\end{theorem}

\begin{proof}
The general adversary bound composes as $\ADVpm(f \compose g) = \ADVpm(f) \ADVpm(g)$ for boolean functions $f$ and $g$~\cite{HoyerLeeSpalek07negativeadv, Reichardt09spanprogram}, so the claim follows from Eq.~\eqnref{e:advboundtight}.  
\end{proof}

The algorithm behind \thmref{t:advboundtight} is substantially simpler than the algorithm for \thmref{t:querycomplexitytightnonbinary}, although its analysis requires slightly more work.  On input $x \in \{0,1\}^n$, 
the algorithm consists of alternating applications of the input oracle~$O_x$---a unitary that maps $\ket{i, b}$ to $\ket{i, x_i \oplus b}$, for $i = 1, \ldots, n$ and $b \in \{0,1\}$---and a certain fixed reflection.  The reflection is about the eigenvalue-zero subspace of the adjacency matrix $A_G$ for a graph $G$ derived from a dual SDP for $\ADVpm$.  This structure is based on Ambainis's AND-OR formula-evaluation algorithm~\cite{ambainis07nand}.  A previous algorithm in Szegedy's quantum walk model~\cite{Szegedy04walkfocs} ran in $O(\ADVpm(f) \norm{\abst(A_G)})$ queries, where $\norm{\abst(A_G)}$ is the operator norm of the entry-wise absolute value of $A_G$~\cite[Prop.~9.4]{Reichardt09spanprogram}.  Ambainis's approach efficiently removes the dependence on higher-energy portions of the adjacency matrix~$A_G$.  The analysis of the algorithm needs to transfer an ``effective" spectral gap for the adjacency matrix of a related graph into an effective spectral gap for the product of $O_x$ and the fixed reflection.  

In fact, the input oracle is itself a reflection, $O_x^2 = \identity$.  Therefore the algorithm consists of alternating two fixed reflections, much like Grover's search algorithm~\cite{Grover96search}.  It follows that every boolean function can be evaluated, with bounded error, optimally in this way.  While known algorithms can in principle be converted into this form~\cite[Theorems~3.1,~5.2]{Reichardt09spanprogram}, we do not know an explicit closed form for the second reflection, e.g., for the collision problem.  

The bipartite graph~$G$ can be thought of as a span program~\cite{KarchmerWigderson93span}, and was constructed in the span program framework of~\cite{Reichardt09spanprogram}.  Thus our algorithm is naturally seen as a quantum algorithm for evaluating span programs.  Since the best span program for a function has witness size exactly equal to the general adversary bound~\cite{Reichardt09spanprogram, Reichardt10wsize}, \thmref{t:advboundtight} also implies that quantum computers, measured by query complexity, and span programs, measured by witness size, are equivalent computational models for evaluating boolean functions.  For simplicity, though, we will not detail this connection further.  We will require from~\cite{Reichardt09spanprogram} only \thmref{t:bipartitepsdreduction} below, which without reference to span programs gives an effective spectral gap for a bipartite graph.  

Barnum, Saks and Szegedy~\cite{BarnumSaksSzegedy03adv} have given a family of SDPs that characterize quantum query complexity according to their feasibility or infeasibility, instead of according to the optimum value of a single SDP.  The BSS SDPs work for any specified error rate, including zero.  The general adversary bound is a polynomially smaller SDP, but of course the truth table of a function is typically exponentially long.  Whereas our algorithm uses a workspace of $n + O(\log n)$ qubits to evaluate an $n$-bit boolean function (by~\cite[Lemma~6.6]{Reichardt09spanprogram}), $n+1$ qubits suffice by~\cite{BarnumSaksSzegedy03adv}.  To the author's knowledge, neither the BSS SDPs nor $\ADVpm$ have ever been solved directly for a nontrivial, asymptotically large family of functions, with better than a constant-factor improvement over the adversary bound (see~\cite{ChildsLee07orderedsearch}).  However, the easy composition rule for $\ADVpm$, used in \thmref{t:composition}, allows for computing $\ADVpm$ for a read-once formula by multiplying the bounds computed for constant-size gates.  It may be that the very simple form of our algorithm will allow for further progress in the understanding and development of quantum query algorithms.

\subsection{Open problems}

An appealing open problem in quantum computing is to show a tighter relationship between classical and quantum query complexities for total functions---the largest known gap is quadratic but the best upper bound is $D(f) = O(Q(f)^6)$~\cite{BealsBuhrmanCleveMoscaWolf98}.  Speculatively, the strong composition properties of quantum algorithms for total boolean functions may be a tool for approaching this problem.  It also remains interesting to study non-boolean functions, their composition and the necessity of the $\log \abs{C}$ and $\log \abs{E}$ factors in Eq.~\eqnref{e:querycomplexityexactlytightnonbinary}.  \thmref{t:composition}, the two-reflections form of the algorithm, and the elimination of the $\log \ADVpm(f)$ factor suggest that it may be possible to adapt the algorithm to evaluate any boolean function $f$ with a \emph{bounded-error} input oracle with the same asymptotic number $\Theta(\ADVpm(f))$ of quantum queries, following~\cite{HoyerMoscaDeWolf03berror-search} for the OR function.  Classically, in the noisy decision-tree model, an extra logarithmic factor for error reduction is sometimes required~\cite{FeigeRaghavanPelegUpfal94noisy}, but this factor is not known to be needed for any quantum query algorithm~\cite{BuhrmanNewmanRohrigDeWolf05robust}.

\subsection{Definitions} \label{s:definitions}

For a natural number $n \in \N$, let $[n] = \{1, 2, \ldots, n\}$.  For a bit $b \in \{0,1\}$, let $\bar b = 1-b$.  For a finite set~$X$, let $\C^X$ be the Hilbert space $\C^{\abs X}$ with orthonormal basis $\{ \ket x : x \in X \}$.  For vector spaces~$V$ and $W$ over $\C$, let $\L(V, W)$ denote the set of all linear transformations from $V$ into $W$, and let~$\L(V) = \L(V, V)$.  $\norm{A}$ is the spectral norm of an operator $A$.

\section{The algorithms} \label{s:algorithm}

For a boolean function, taking the dual of the general adversary bound SDP in \defref{t:adversarydef} gives: 

\begin{lemma}[{\cite[Theorem~6.2]{Reichardt09spanprogram}}] \label{t:spanprogramSDPconstruction}
Let $f : \D \rightarrow \{0,1\}$, with $\D \subseteq \{0,1\}^n$.  For $b \in \{0,1\}$, let $F_b = \{ x \in \D : f(x) = b \}$.  Then 
\begin{equation}\begin{split} \label{e:spanprogramSDPCholesky}
\ADVpm(f)
&=
\min_{\Large \substack{m \in \N, \{ \ket{v_{x j}} \in \C^m : \, x \in \D, j \in [n] \} \,: \\ \forall (x,y) \in F_0 \times F_1, \, \sum_{j \in [n] : x_j \neq y_j} \braket{v_{x j}}{v_{y j}} = 1}} \max_{\Large \substack{x \in \D}} \; \sum_{j \in [n]} \norm{ \ket{v_{x j}} }^2
 \enspace .
\end{split}\end{equation}
\end{lemma}
\noindent
Based on a feasible solution to this SDP with objective value $W (\geq 1)$, we will give three algorithms for evaluating~$f$, each with query complexity $O(W)$.  (A feasible solution corresponds to a span program in canonical form, and its value equals the span program witness size~\cite{KarchmerWigderson93span, Reichardt09spanprogram}.)  

Let $I = [n] \times \{0,1\} \times [m]$.  Let $\ket t \in \C^{F_0}$ and $A \in \L(\C^{F_0}, \C^I)$ be given by 
\begin{equation}\begin{split} \label{e:tAdefs}
\ket t &= \frac 1 {3 \sqrt W} \sum_{x \in F_0} \ket x \\
A &= \sum_{x \in F_0, j \in [n]} \ketbra{x}{j, \bar x_j} \otimes \bra{v_{x j}}
 \enspace .
\end{split}\end{equation}
Let $G$ be the weighted bipartite graph with biadjacency matrix $\biadj_G \in \L(\C^{\{0\}} \oplus \C^I, \C^{F_0})$: 
\begin{equation} \label{e:Gadjacencymatrix}
\biadj_G = \left( \begin{matrix} \ket t & A \end{matrix} \right)
 \enspace .
\end{equation}
That is, $G$ has a vertex for each row or column of~$\biadj_G$; its vertex set is the disjoint union~$F_0 \cup \{0\} \cup I$.  Edges from~$F_0$ to $\{0\} \cup I$ are weighted by the matrix entries.  The weighted adjacency matrix of $G$ is 
\begin{equation}
A_G = \left( \begin{matrix} 0 & \biadj_G \\ \biadj_G^\dagger & 0 \end{matrix} \right)
 \enspace .
\end{equation}

Let $\Delta \in \L(\C^{F_0 \cup \{0\} \cup I})$ be the orthogonal projection onto the span of all eigenvalue-zero eigenvectors of~$A_G$.  For an input $x \in \D$, let $\Pi_x \in \L(\C^{F_0 \cup \{0\} \cup I})$ be the projection 
\begin{equation}
\Pi_x = \identity - \sum_{j \in [n], k \in [m]} \ketbra{j, \bar x_j, k}{j, \bar x_j, k}
 \enspace .
\end{equation}
That is, $\Pi_x$ is a diagonal matrix that projects onto all vertices except those associated to the input bit complements~$\bar x_j$.  Finally, let 
\begin{equation}
U_x = (2 \Pi_x - \identity) (2 \Delta - \identity)
 \enspace .
\end{equation}
$U_x$ consists of the alternating reflections $2 \Delta - \identity$ and $2 \Pi_x - \identity$.  The first reflection does not depend on the input~$x$.  The second reflection can be implemented using a single call to the input oracle~$O_x$.  

We present three related algorithms, each slightly simpler than the one before: 

\def\algbox#1{\begin{center}\fbox{ \begin{minipage}[l]{5.50in}#1\end{minipage} }\end{center}}

\algbox{
\noindent {\bf Algorithm 1:}
\begin{enumerate}
\item Prepare the initial state $\ket 0 \in \C^{F_0 \cup \{0\} \cup I}$.  
\item Run phase estimation on $U_x$, with precision $\delta_p = \frac 1 {100W}$ and error rate $\delta_e = \frac{1}{10}$.  
\item Output $1$ if the measured phase is zero.  Otherwise output $0$.
\end{enumerate}
}

\algbox{
\noindent {\bf Algorithm 2:}
\begin{enumerate}
\item Prepare the initial state $\frac{1}{\sqrt 2}(\ket 0 + \ket 1) \otimes \ket 0 \in \C^2 \otimes \C^{F_0 \cup \{0\} \cup I}$.  
\item Pick $T \in [ \lceil 100W \rceil ]$ uniformly at random.  Apply the controlled unitary $\ketbra 0 0 \otimes \identity + \ketbra 1 1 \otimes U_x^T$.  
\item Measure the first register in the basis $\frac{1}{\sqrt 2}(\ket 0 \pm \ket 1)$.  Output $1$ if the measurement result is $\frac{1}{\sqrt 2}(\ket 0 + \ket 1)$, and output $0$ otherwise.  
\end{enumerate}
}

\algbox{
\noindent {\bf Algorithm 3:}
\begin{enumerate}
\item Prepare the initial state $\ket 0 \in \C^{F_0 \cup \{0\} \cup I}$.  
\item Pick $T \in [ \lceil 10^5 W \rceil ]$ uniformly at random.  Apply $U_x^T$.  
\item Measure the vertex.  Output $1$ if the measurement result is $\ket 0$, and output $0$ otherwise.
\end{enumerate}
}

Phase estimation on a unitary $V$ with precision $\delta_p$ and error rate $\delta_e$ can be implemented using $O\big(\frac 1 {\delta_p} \log \frac 1 {\delta_e}\big)$ controlled applications of $V$~\cite{NagajWocjanZhang09qma}, so the first algorithm has $O(W)$ query complexity.  The second algorithm essentially applies a simplified version of phase estimation.  Intuitively, it works because it suffices to distinguish zero from nonzero phase.  The third algorithm does away with any phase estimation.  Intuitively, this is possible because $U_x$ is the product of two reflections, so its spectrum is symmetrical.  The second and third algorithms clearly have $O(W)$ query complexity.  The factor of $10^5$ in the third algorithm's query complexity can be reduced by three orders of magnitude by adjusting downward the scaling factor for $\ket t$ in Eq.~\eqnref{e:tAdefs}.  

The time, or number of elementary operations, required to implement the reflection $2 \Delta - \identity$ is unclear.  In practice it may still be preferable to use the potentially less query-efficient quantum walk algorithm from~\cite{Reichardt09spanprogram}, as done for evaluating formulas in~\cite{Reichardt09unbalancedformula, ReichardtSpalek08spanprogram, Reichardt09andorfaster, AmbainisChildsReichardtSpalekZhang07andor}.   

In the following section, we will show that all three algorithms correctly evaluate $f(x)$, with constant gaps between the soundness error and completeness parameters.

\section{Analysis of the algorithms}

To analyze the above algorithms, we shall study the spectrum of the unitary $U_x = (2 \Pi_x - \identity) (2 \Delta - \identity)$.  

For this purpose, it will be useful to introduce two new graphs, following~\cite{ReichardtSpalek08spanprogram, Reichardt09spanprogram}.  Let~$\Pim(x) = \sum_{j \in [n]} \ketbra j j \otimes \ketbra{\bar x_j}{\bar x_j} \otimes \identity_{\C^{[m]}} \in \L(\C^I)$, and let $G(x)$ and $G'(x)$ be the weighted bipartite graphs with biadjacency matrices 
\begin{equation} \label{e:Gxadjacencymatrix}
B_{G(x)} = \left( \begin{matrix} \ket t & A \\ 0 & \Pim(x) \end{matrix} \right)
\qquad \text{and} \qquad
B_{G'(x)} = \left( \begin{matrix} A \\ \Pim(x) \end{matrix} \right)
 \enspace .
\end{equation}

Based on the constraints of the SDP in \lemref{t:spanprogramSDPconstruction}, we can immediately construct eigenvalue-zero eigenvectors for $G(x)$ or $G'(x)$, depending on whether $f(x) = 1$ or $f(x) = 0$: 

\begin{lemma} \label{t:eigenvaluezeroGx}
If $f(x) = 1$, let $\ket \psi = -3 \sqrt W \ket 0 + \sum_{j \in [n]} \ket{j, x_j} \otimes \ket{v_{xj}} \in \C^{\{0\} \cup I}$.  Then $B_{G(x)} \ket \psi = 0$ and $\abs{ \braket 0 \psi }^2 / \norm{\ket \psi}^2 \geq 9/10$.  

If $f(x) = 0$, let $\ket \psi = -\ket x + \sum_{j \in [n]} \ket{j, \bar x_j} \otimes \ket{v_{xj}} \in \C^{F_0 \cup I}$.  Then $B_{G'(x)}^\dagger \ket \psi = 0$ and $\abs{ \braket t \psi }^2 / \norm{\ket \psi}^2 \geq 1 / (9 W (W + 1))$.  
\end{lemma}

Let us recall from~\cite{Reichardt09spanprogram}: 

\begin{theorem}[{\cite[Theorem~8.7]{Reichardt09spanprogram}}] \label{t:bipartitepsdreduction}
Let $G'$ be a weighted bipartite graph with biadjacency matrix $\biadj_{G'} \in \L(\C^U, \C^T)$.  Assume that $\delta > 0$ and $\ket t, \ket \psi \in \C^T$ satisfy $B_{G'}^\dagger \ket \psi = 0$ and $\abs{\braket t \psi}^2 \geq \delta \norm{\ket \psi}^2$.  

Let $G$ be the same as $G'$ except with a new vertex, $0$, added to the $U$ side, with outgoing edges weighted by the entries of $\ket t$.  That is, the biadjacency matrix of $G$ is 
\begin{equation}
\biadj_G = \left( \begin{matrix} \ket t & \biadj_{G'} \end{matrix} \right)
\place{{\small $0$}}{-65mu}{12pt}
\place{{\small U}}{-27mu}{12pt}
\place{{\small T}}{3mu}{0pt}
\end{equation}
Let $\{ \ket \alpha \}$ be a complete set of orthonormal eigenvectors of the weighted adjacency matrix $A_G$, with corresponding eigenvalues $\rho(\alpha)$.  Then for all $\gamma \geq 0$, the squared length of the projection of $\ket 0$ onto the span of the eigenvectors $\alpha$ with $\abs{\rho(\alpha)} \leq \gamma$ satisfies 
\begin{equation} \label{e:bipartitepsdreduction}
\sum_{\alpha : \, \abs{\rho(\alpha)} \leq \gamma} \abs{\braket \alpha 0}^2 \leq 8 \gamma^2 / \delta
 \enspace .
\end{equation}
\end{theorem}

Substituting \lemref{t:eigenvaluezeroGx} into \thmref{t:bipartitepsdreduction}, we thus obtain the key statement: 

\begin{lemma} \label{t:effectivespectralgap}
If $f(x) = 1$, then $A_{G(x)}$ has an eigenvalue-zero eigenvector $\ket \psi$, supported on the column vertices, with 
\begin{equation} \label{e:eigenvaluezero}
\frac{ \abs{ \braket 0 \psi }^2 }{\norm{\ket \psi}^2} \geq \frac{9}{10}
 \enspace .
\end{equation}

If $f(x) = 0$, let $\{ \ket \alpha \}$ be a complete set of orthonormal eigenvectors  of $A_{G(x)}$ with corresponding eigenvalues $\rho(\alpha)$.  Then for any $c \geq 0$, 
\begin{equation} \label{e:effectivespectralgap}
\sum_{\alpha : \abs{\rho(\alpha)} \leq c/W} \abs{ \braket \alpha 0 }^2 \leq 72 \Big(1 + \frac 1 W\Big) c^2
 \enspace .
\end{equation}
\end{lemma}

By choosing $c$ a small positive constant, Eq.~\eqnref{e:effectivespectralgap} gives an $O(1/W)$ ``effective spectral gap" for eigenvectors of $A_{G(x)}$ supported on $\ket 0$; it says that $\ket 0$ has small squared overlap on the subspace of~$O(1/W)$-eigenvalue eigenvectors.  

So far, we have merely repeated arguments from~\cite{Reichardt09spanprogram}.  The main step in the analysis of our new algorithms is to translate \lemref{t:effectivespectralgap} into analogous statements for $U_x$: 

\begin{lemma} \label{t:completenesssoundness}
If $f(x) = 1$, then $U_x$ has an eigenvalue-one eigenvector $\ket \varphi$ with 
\begin{equation}
\frac{ \abs{ \braket 0 \varphi }^2 }{\norm{\ket \varphi}^2} \geq \frac{9}{10}
 \enspace .
\end{equation}

If $f(x) = 0$, let $\{ \ket \beta \}$ be a complete set of orthonormal eigenvectors  of $U_x$ with corresponding eigenvalues $e^{i \theta(\beta)}$, $\theta(\beta) \in (-\pi, \pi]$.  Then for any $\Theta \geq 0$, 
\begin{equation}
\sum_{\beta : \abs{\theta(\beta)} \leq \Theta} \abs{\braket \beta 0}^2 \leq \Big(2 \sqrt{6 \Theta W} + \frac \Theta 2\Big)^2
 \enspace .
\end{equation}
\end{lemma}

Assuming \lemref{t:completenesssoundness}, the algorithms from \secref{s:algorithm} are both complete and sound.  
If $f(x) = 1$, then the first, phase-estimation-based algorithm outputs $1$ with probability at least $9/10 - \delta_e = 4/5$.  If $f(x) = 0$, then setting $\Theta = \delta_p = \frac 1 {100W}$, the algorithm outputs $1$ with probability at most $\delta_e + (2 \sqrt{6 \Theta W} + \frac \Theta 2)^2  < 2/5$.  The probability the second algorithm outputs $1$ is the expectation versus $T$ of $\frac 1 4 \norm{(\identity + U_x^T) \ket 0}{}^2$.  If $f(x) = 1$, this is at least $9/10$ for all $T$.  If $f(x) = 0$, let $\tau = \lceil 100 W \rceil$ and simplify 
\begin{equation}\begin{split}
\Ex_{T \in_R [\tau]}\!\Big[ \frac 1 4 \norm{(\identity + U_x^T) \ket 0}{}^2 \Big]
&= \Ex_{T \in_R [\tau]}\!\Big[ \frac 1 4 \sum_\beta \abs{1 + e^{i \theta(\beta) T}}{}^2 \abs{\braket 0 \beta}^2 \Big] \\
&= \frac 1 4 \sum_\beta \abs{\braket 0 \beta}^2 \Big[2 + \frac 1 \tau \Big( \frac{e^{i \theta(\beta)(\tau+1)}-e^{-i \theta(\beta)\tau}}{e^{i \theta(\beta)}-1} - 1 \Big) \Big]
 \enspace .
\end{split}\end{equation}
Setting $\Theta = \frac 1 {50W}$ and $\xi = (2 \sqrt{6 \Theta W} + \frac \Theta 2)^2$, we see that the algorithm outputs $1$ with probability at most $\xi + (1-\xi)\big(\frac 1 2 + 1/(4 \tau \sin \frac \Theta 2)\big) < 88\%$ for $W \geq 1$.  
As its analysis requires more care, we defer consideration of the third algorithm to the end of this section.  

\smallskip

For the proof of \lemref{t:completenesssoundness} we will use the following characterization of the eigen-decomposition of the product of reflections, essentially due to Jordan~\cite{Jordan75projections}.  Its use is common in quantum computation, e.g.,~\cite{NagajWocjanZhang09qma, Szegedy04walkfocs, MarriottWatrous05qma}.  

\begin{lemma}
Given two projections $\Pi$ and $\Delta$, the Hilbert space can be decomposed into orthogonal one- and two-dimensional subspaces invariant under $\Pi$ and $\Delta$.  On the one-dimensional invariant subspaces, $(2 \Pi - \identity)(2 \Delta - \identity)$ acts as either $+1$ or $-1$.  Each two-dimensional subspace is spanned by an eigenvalue-$\lambda$ eigenvector $\ket v$ of $\Delta \Pi \Delta$, with $\lambda \in (0,1)$, and $\ket{v^\perp} = (\identity - \Delta) \Pi \ket v / \norm{(\identity - \Delta) \Pi \ket v}$.  Letting $\theta = 2 \arccos \sqrt \lambda \in (0, \pi)$, 
so $\Pi \ket v / \norm{\Pi \ket v} = \cos \frac \theta 2 \ket v + \sin \frac \theta 2 \ket {v^\perp}$, 
the eigenvectors and corresponding eigenvalues of $(2 \Pi - \identity)(2 \Delta - \identity)$ on this subspace are, respectively,  
\begin{equation}
\frac{ \ket v \mp i \ket{v^\perp} }{\sqrt 2} \qquad \text{and} \qquad e^{\pm i \theta}
 \enspace .
\end{equation}
\end{lemma}

\begin{proof}[Proof of \lemref{t:completenesssoundness}]
Notice from Eqs.~\eqnref{e:Gadjacencymatrix} and~\eqnref{e:Gxadjacencymatrix} that $G$ is naturally a subgraph of $G(x)$.  Since~$A_G \Delta = 0$ by definition of $\Delta$, $A_{G(x)} \Delta = T (\identity - \Pi_x)$, where $T$ is a permutation matrix.  

First consider the case $f(x) = 1$.  Let $\ket \varphi$ be the restriction of $\ket \psi$ from Eq.~\eqnref{e:eigenvaluezero} to the vertices of~$G$.  Since $\ket \psi$ has no support on the extra vertices of $G(x)$, $\norm{\ket \varphi} = \norm{\ket \psi}$ and $\ket \varphi$ is an eigenvalue-zero eigenvector of $A_G$; $\Delta \ket \varphi  = \ket \varphi$.  Also $\Pi_x \ket \varphi = \ket \varphi$, so indeed $U_x \ket \varphi = \ket \varphi$.  

\def\tz{\zeta}
\def\nz{\hat{\tz}}

Next consider the case $f(x) = 0$.  Let 
\begin{equation}
\ket{\tz} = \sum_{\beta : \abs{\theta(\beta)} \leq \Theta} \ket \beta \braket \beta 0
\end{equation}
be the projection of $\ket 0$ onto the space of eigenvectors with phase at most $\Theta$ in magnitude.  Our aim is to upper bound $\norm{\ket \tz}{}^2 = \braket 0 \tz = \abs{\braket 0 \nz}{}^2$, where $\ket \nz = \ket \tz/\norm{\ket \tz}$.  Notice that $\ket \nz$ is supported only on eigenvectors $\ket \beta$ with $\theta(\beta) \neq 0$, i.e., on the two-dimensional invariant subspaces of $\Pi_x$ and~$\Delta$.  Indeed, if $U_x \ket \beta = \ket \beta$, then either $\ket \beta = \Pi_x \ket \beta = \Delta \ket \beta$ or $\ket \beta = (\identity - \Pi_x) \ket \beta = (\identity - \Delta) \ket \beta$.  The first possibility implies $A_{G(x)} \ket \beta = 0$, so by Eq.~\eqnref{e:effectivespectralgap} with $c = 0$, $\braket 0 \beta = 0$.  In the second possibility, also $\braket 0 \beta = \bra 0 \Pi_x \ket \beta = 0$ since $\ket 0 = \Pi_x \ket 0$.  

We can split $\braket 0 \nz$ as  
\begin{equation}\begin{aligned}
\braket 0 \nz 
&= \bra 0 \Delta \ket \nz + \bra 0 \Pi_x (\identity-\Delta) \ket \nz \\
&\leq \abs{\bra 0 \Delta \ket \nz} + \abs{\bra 0 \Pi_x (\identity-\Delta) \ket \nz} \\
&\leq \abs{\bra 0 \Delta \ket \nz} + \norm{\Pi_x (\identity-\Delta) \ket \nz}
 \enspace .
\end{aligned}\end{equation}

Start by bounding the second term, $\norm{\Pi_x (\identity-\Delta) \ket \nz}$.  Intuitively, this term is small because $\ket \nz$ is supported only on two-dimensional invariant subspaces where $\Delta$ and $\Pi_x$ are close.  Indeed, 
let $\ket{{-\beta}} = (2\Delta-\identity) \ket \beta$, an eigenvector of $A_G$ with phase $\theta(-\beta) = -\theta(\beta)$.  
Expanding $\ket \nz = \sum_\beta c_\beta \ket \beta$, 
\begin{align}
\norm{\Pi_x (\identity-\Delta) \ket \nz}{}^2 \nonumber
&= \norm{ \sum_\beta  \Pi_x (\identity - \Delta) c_\beta \ket \beta }{}^2 \\
&= \sum_{\beta : \theta(\beta) > 0} \norm{ \Pi_x (\identity - \Delta) (c_\beta \ket \beta + c_{-\beta} \ket{{-\beta}}) }{}^2 \nonumber \\
&= \sum_{\beta : \theta(\beta) > 0} \sin^2 \frac{\theta(\beta)}{2} \, \norm{(\identity - \Delta)(c_\beta \ket \beta + c_{-\beta} \ket{{-\beta}})}^2 \nonumber \\
&\leq \Big(\frac \Theta 2\Big)^2 \norm{(\identity - \Delta)\ket \nz}{}^2
 \enspace .
\end{align}

It remains to bound $\abs{\bra 0 \Delta \ket \nz} = \abs{\braket 0 w} \norm{\Delta \ket \nz}$, where $\ket w = \Delta \ket \nz / \norm{\Delta \ket \nz}$ is an eigenvalue-zero eigenvector of $A_G$.  Intuitively, if $\abs{\braket 0 w} = \abs{\bra 0 \Pi_x \ket w}$ is large, then since $A_G$ and $A_{G(x)}$ are the same on $\Pi_x$, $\norm{A_{G(x)} \ket w} = \norm{T (\identity - \Pi_x) \ket w}$ will be small.  This in turn will imply that $\ket w$ has large support on the small-eigenvalue subspace of $A_{G(x)}$, contradicting Eq.~\eqnref{e:effectivespectralgap}.  

Beginning the formal argument, we have 
\begin{equation}\begin{split}
\norm{A_{G(x)} \Delta \ket \nz}{}^2
&= \norm{(\identity - \Pi_x) \Delta \ket \nz}{}^2 \\
&= \sum_{\beta : \theta(\beta) > 0} \norm{(\identity-\Pi_x) \Delta (c_\beta \ket \beta + c_{-\beta} \ket{{-\beta}})}^2 \\
&= \sum_{\beta : \theta(\beta) > 0} \sin^2 \frac{\theta(\beta)}{2} \, \norm{\Delta (c_\beta \ket \beta + c_{-\beta} \ket{{-\beta}})}^2 \\
&\leq \Big(\frac \Theta 2\Big)^2 \norm{\Delta \ket \nz}{}^2
 \enspace .
\end{split}\end{equation}
Hence $\norm{A_{G(x)} \ket w} \leq \Theta / 2$.  

\def\vgood{w_\text{small}}
\def\vbad{w_\text{big}}
Now split $\ket w = \ket \vgood + \ket \vbad$, where for a certain $d > 0$ to be determined, 
\begin{equation}\begin{split}
\ket \vgood &= \sum_{\alpha : \abs{\rho(\alpha)} \leq d \Theta/2} \ket \alpha \braket \alpha w \\
\ket \vbad &= \sum_{\alpha : \abs{\rho(\alpha)} > d \Theta/2} \ket \alpha \braket \alpha w
 \enspace .
\end{split}\end{equation}
Then  
\begin{equation}
\abs{\bra 0 \Delta \ket \nz}
\leq \abs{\braket 0 w} 
\leq \abs{\braket 0 \vgood} + \abs{\braket 0 \vbad} 
 \enspace .
\end{equation}

From Eq.~\eqnref{e:effectivespectralgap} with $c = d \Theta W/2$, $\abs{\braket 0 \vgood} \leq \sqrt{72 (1+1/W)} c \norm{\ket \vgood} \leq 6 d \Theta W$.  

Since $A_{G(x)} \ket w = \sum_\alpha \rho(\alpha) \ket \alpha \braket \alpha w$, we have 
\begin{equation}\begin{split}
\Big(\frac \Theta 2\Big)^2
&\geq \norm{A_{G(x)} \ket w}^2 \\
&= \norm{A_{G(x)} \ket \vgood}^2 + \norm{A_{G(x)} \ket \vbad}^2 \\
&\geq d^2 \Big(\frac \Theta 2\Big)^2 \norm{\ket \vbad}^2
 \enspace .
\end{split}\end{equation}
Hence $\norm{\ket \vbad} \leq 1/d$.  

Combining our calculations gives 
\begin{equation}\begin{split}
\sqrt{\sum_{\beta : \abs{\theta(\beta)} \leq \Theta} \abs{\braket \beta 0}^2}
= \braket 0 \nz 
\leq \abs{\bra 0 \Delta \ket \nz} + \norm{\Pi_x (\identity - \Delta) \ket \nz} 
\leq 6 d \Theta W + \frac 1 d + \frac \Theta 2
 \enspace .
\end{split}\end{equation}
The right-hand side is $2 \sqrt{6 \Theta W} + \Theta/2$, as claimed, for $d = 1/\sqrt{6 \Theta W}$.  
\end{proof}

Having proved \lemref{t:completenesssoundness}, we return to the correctness proof for the third algorithm.  

\begin{proposition}
If $f(x) = 1$, then the third algorithm outputs $1$ with probability at least $64\%$.  If~$f(x) = 0$, then the third algorithm outputs $1$ with probability at most $61\%$.  
\end{proposition}

\begin{proof}
Letting $\tau = \lceil 10^5 W \rceil$, the third algorithm outputs $1$ with probability 
\begin{equation} \label{e:thirdalgorithmoutputsone}
p_1 := \Ex_{T \in_R [\tau]}\!\big[\abs{\bra 0 U_x^T \ket 0}{}^2\big] = \Ex_{T \in_R [\tau]} \Big\vert \sum_\beta e^{i \theta(\beta) T} \abs{\braket \beta 0}^2 \Big\vert^2
 \enspace .
\end{equation}
If $f(x) = 1$, then a crude bound puts $p_1$ at least $(9/10 - 1/10)^2 = 64\%$.  

Assume $f(x) = 0$.  Recall the notation that for an eigenvector $\ket \beta$ with $\abs{\theta(\beta)} \in (0, \pi)$, $\ket{{-\beta}} = (2 \Delta - \identity) \ket \beta$ denotes the corresponding eigenvector with eigenvalue phase $\theta(-\beta) = - \theta(\beta)$.  The key observation for this proof is that 
\begin{equation}
\braket 0 \beta = e^{-i \theta(\beta)} \braket 0 {{-\beta}}
 \enspace .
\end{equation}
This equal splitting of $\abs{\braket 0 \beta}$ and $\abs{\braket 0 {{-\beta}}}$ will allow us to bound $p_1$ close to $1/2$ instead of the trivial bound $p_1 \leq 1$.  The intuition is that after applying $U_x$ a suitable number of times, eigenvectors~$\ket \beta$ and $\ket {{-\beta}}$ will accumulate roughly opposite phases, so their overlaps with $\ket 0$ will roughly cancel out.  For this argument to work, though, the eigenvalue phase $\theta(\beta)$ should be bounded away from zero and from $\pm \pi$.  Therefore define the projections 
\begin{equation}\begin{split}
\Delta_\Theta &= \sum_{\beta : \abs{\theta(\beta)} \leq \Theta} \ketbra \beta \beta \\
\overline{\Delta}_\Lambda &= \sum_{\beta : \abs{\theta(\beta)} > \Lambda} \ketbra \beta \beta \\
\Sigma &= \identity - \Delta_\Theta - \overline{\Delta}_\Lambda
 \enspace ,
\end{split}\end{equation}
where $\Theta$ and $\Lambda$, $0 < \Theta < \Lambda < \pi$, will be determined below.  \lemref{t:completenesssoundness} immediately gives the bound~$\norm{\Delta_\Theta \ket 0} \leq 2 \sqrt{6 \Theta W} + \frac \Theta 2$.  We can also place a bound on $\norm{\overline{\Delta}_\Lambda \ket 0}$, using 
\begin{equation}
2(\Delta - \identity) \ket 0 = (U_x^\dagger - \identity) \ket 0 = \sum_\beta (e^{-i \theta(\beta)} - 1) \ket \beta \braket \beta 0
 \enspace .
\end{equation}
Expanding the squared norm of both sides gives  
\begin{align}
\norm{(U_x^\dagger - \identity) \ket 0}{}^2
&= 
4 \sum_\beta \sin^2 \frac{\theta(\beta)}{2} \abs{\braket \beta 0}^2
\geq 
\norm{\overline{\Delta}_\Lambda \ket 0}{}^2 \cdot 4 \sin^2 \frac \Lambda 2 \label{e:minusonegap1}
\intertext{and}
\norm{(U_x^\dagger - \identity) \ket 0}{}^2
&=
4 (1 - \norm{\Delta \ket 0}^2) \leq 2/5 \label{e:minusonegap2}
 \enspace .
\end{align}
In the second step we have used that $\norm{\Delta \ket 0}^2 \geq 9/10$; provided that $f$ is not the constant zero function, $A_G$ must have an eigenvalue-zero eigenvector with large overlap on $\ket 0$.  Combining Eqs.~\eqnref{e:minusonegap1} and~\eqnref{e:minusonegap2} gives 
\begin{equation}
\norm{\overline{\Delta}_\Lambda \ket 0}{}^2 \leq \frac{1}{10 \sin^2 \frac \Lambda 2}
 \enspace .
\end{equation}

Returning to Eq.~\eqnref{e:thirdalgorithmoutputsone}, we have 
\begin{equation}\begin{split} \label{e:thirdalgorithmoutputsone2}
p_1
&\leq
\Ex_{T \in_R [\tau]} \Big( \norm{\Delta_\Theta \ket 0}^2 + \norm{\overline{\Delta}_\Lambda \ket 0}{}^2 + \Big\vert \sum_{\beta : \theta(\beta) \in (\Theta, \Lambda]} \abs{\braket \beta 0}^2 \big(e^{i \theta(\beta) T} + e^{- i \theta(\beta)T}\big) \Big\vert \Big)^2 \\
&\leq 
(\norm{\Delta_\Theta \ket 0}^2 + \norm{\overline{\Delta}_\Lambda \ket 0}{}^2)(2 - \norm{\Delta_\Theta \ket 0}^2 - \norm{\overline{\Delta}_\Lambda \ket 0}{}^2) \\
&\quad + \Ex_{T \in_R [\tau]} \Big( \sum_{\beta : \theta(\beta) \in (\Theta, \Lambda]} \abs{\braket \beta 0}^2 \big(e^{i \theta(\beta) T} + e^{- i \theta(\beta)T}\big) \Big)^2
 \enspace .
\end{split}\end{equation}
The algorithm chooses $T$ at random to allow bounding the last term.  Expanding this term gives 
\begin{equation}\begin{split}
\Ex&_{T \in_R [\tau]} \sum_{\beta, \beta' : \theta(\beta), \theta(\beta') \in (\Theta, \Lambda]} \abs{\braket \beta 0}^2 \abs{\braket{\beta'} 0}^2 \big(e^{i \theta(\beta) T} + e^{- i \theta(\beta)T}\big) \big(e^{i \theta(\beta') T} + e^{- i \theta(\beta')T}\big) \\
&= 
\Ex \sum_{\theta, \theta' \in (\Theta, \Lambda]} \abs{\braket \beta 0}^2 \abs{\braket{\beta'} 0}^2 \big( (e^{i (\theta+\theta') T} + e^{- i (\theta + \theta') T}) + (e^{i (\theta - \theta') T} + e^{- i (\theta - \theta')T})\big) \\
&\leq \frac12 \norm{\Sigma \ket 0}^4 + \Ex \sum_{\theta, \theta' \in (\Theta, \Lambda]} \abs{\braket \beta 0}^2 \abs{\braket{\beta'} 0}^2 \big(e^{i (\theta+\theta') T} + e^{- i (\theta + \theta') T}\big) \\
&= \frac12 \norm{\Sigma \ket 0}^4 + \frac 1 \tau \sum_{\theta, \theta' \in (\Theta, \Lambda]} \abs{\braket \beta 0}^2 \abs{\braket{\beta'} 0}^2 \Big( \frac{e^{i (\theta + \theta') (\tau+1)} - e^{-i (\theta + \theta')\tau}}{e^{i (\theta + \theta')} - 1} - 1\Big) \\
&\leq 
\frac12 \Big(1 + \frac{1/\tau}{\min_{\theta, \theta' \in (\Theta, \Lambda]} \abs{e^{i(\theta+\theta')}-1}}\Big) \norm{\Sigma \ket 0}^4 \\
&\leq 
\frac12 \Big(1 + \frac{1}{2 \tau \min \{\sin \Theta, \sin \Lambda\}}\Big) \norm{\Sigma \ket 0}^4
 \enspace .
\end{split}\end{equation}
Here for brevity we have written $\theta$ and $\theta'$ for $\theta(\beta)$ and $\theta(\beta')$, respectively.  In the second and fourth steps, we have used $\sum_{\theta \in (\Theta, \Lambda]} \abs{\braket \beta 0}^2 = \frac12 \norm{\Sigma \ket 0}^2$.  In the last step, we have used $\abs{e^{i (\theta + \theta')} - 1} = 2 \sin \frac{\theta+\theta'}{2} \geq 2 \min \{\sin \Theta, \sin \Lambda\}$.  
Substituting the result back into Eq.~\eqnref{e:thirdalgorithmoutputsone2} gives
\begin{equation}\begin{split}
p_1 
&\leq
1 - \frac12 \Big(1 - \frac{1}{2 \tau \min \{\sin \Theta, \sin \Lambda\}}\Big) \norm{\Sigma \ket 0}^4 \\
&\leq 
1 - \frac12 \Big(1 - \frac{1}{2 \tau \min \{\sin \Theta, \sin \Lambda\}}\Big) \max \Big[1 - \frac{1}{10 \sin^2 \frac \Lambda 2} - \Big(2 \sqrt{6 \Theta W} + \frac \Theta 2\Big)^2, 0 \Big]^2
 \enspace .
\end{split}\end{equation}
Setting $\Lambda = \pi - \Theta$ and $\Theta = 1/ (2000W)$, for $W \geq 1$ a calculation yields $p_1 \leq 61\%$.  
\end{proof}

\subsection*{Acknowledgements}
I would like to thank Troy Lee for his help in formulating \thmref{t:querycomplexitytightnonbinary}.  I also thank Sergio Boixo, Stephen Jordan, Julia Kempe and Rajat Mittal for helpful comments, and the Institute for Quantum Information for hospitality.  Research supported by NSERC, ARO and MITACS.

\bibliographystyle{alpha-eprint}
\bibliography{andor}

\bigskip\bigskip
\indent Institute for Quantum Computing, University of Waterloo\\
\indent E-mail address: {\tt breic@iqc.ca}

\end{document}